\documentclass{amsart}

\usepackage{amsmath}
\usepackage{amssymb}
\usepackage{amsthm}
\usepackage{amsfonts}
\usepackage{gregory_pack}

\usepackage[utf8]{inputenc}

\newcommand{\RV}{{X}}
\newcommand{\RP}{{\XX}}
\newcommand{\Mr}{{M_{rel}}}
\newcommand{\Info}{\ic}
\newcommand{\Unif}{\uc}
\newcommand{\Mean}{{M}}
\newcommand{\un}{^{(n)}}

\newcommand{\Div}{{D}}
\newcommand{\Divkl}{{D_{KL}}}
\newcommand{\Dr}{{D_{rel}}}

\title{On typical encodings of multivariate ergodic sources} 
\begin{document}
\dedicatory{In memory of Fero Matúš}
  
\begin{abstract}
  We show that the typical coordinate-wise encoding of multivariate ergodic source into prescribed alphabets has the entropy profile close to the convolution of the entropy profile of the source and the modular polymatroid that is determined by the cardinalities of the output alphabets. We show that the proportion of the exceptional encodings that are not close to the convolution goes to zero doubly exponentially. The result holds for a class of multivariate sources that satisfy asymptotic equipartition property described via the mean fluctuation of the information functions. This class covers asymptotically mean stationary processes with ergodic mean, ergodic processes, irreducible Markov chains with an arbitrary initial distribution. We also proved that typical encodings yield the asymptotic equipartition property for the output variables. 
  These asymptotic results are based on an explicit lower bound of the proportion of encodings that transform a multivariate random variable into a variable with the entropy profile close to the suitable convolution.
\end{abstract}

\maketitle
\author{
 M. Kupsa
 \footnote{\tiny Institute of Information Theory and
   Automation, The Academy of Sciences of the Czech Republic, Prague 8, CZ-18208,
 \email{kupsa@utia.cas.cz}
}



\section{Introduction}

In Information theory and namely in Network coding theory, the encoding of multiple possibly correlated sources have been extensively studied in the last two decades (for the main framework, overview and references see  
\cite{yeung2008information} and \cite{bassoli2013network}). The sources are usually assumed to be discrete and memory-less, so they are represented by sequences of i.i.d. discrete random variables. The important characteristics of multivariate sources, which we focus on, is its entropy profile (an entropic point corresponding to a given multivariate source, see \cite{Ka12thesis}).

We deal with a problem of how the entropy profile changes when "typical" coordinate-wise encodings into prescribed output alphabets are applied. As was shown in \cite[Theorem 3]{Ma07a}, in an asymptotic case, a typical encoding saves as much of the original information as possible. Namely, the conditional entropy of the encoded variable is naturally bounded from above by the conditional entropy of the source and also by the logarithm of the size of the output alphabet. Matúš showed that this bound is asymptotically tight. More literally,
a typical coordinate-wise encoding preserves almost all conditional entropy whenever the output alphabet is large enough, i.e., when the logarithm of the alphabet size exceeds the conditional entropy. If the logarithm is not larger, the conditional entropy of the encoded variable is close to the logarithm of the alphabet size. This observation is used to prove the closeness of the entropy region under the convolution with modular polymatroids (\cite[Theorem 2]{Ma07a}). The role of convolution in the research on the entropy region was then more explored in \cite{matuvs2016entropy}.    

In \cite{Ma07a}, the coordinate-wise encodings of the original multivariate source into prescribed alphabets are applied inductively, coordinate by coordinate, and in between these inductive steps, one has to pass from a random vector to its i.i.d. expansion (see the proof of Theorem 2 in the discussed paper). In particular, each encoding is applied to a different random vector. This procedure can be reinterpreted as simultaneous coordinate encodings used on one fixed i.i.d. expansion of the original entropy vector. But this simple reasoning does not allow to deduce that the entropy profile obtained for a specific encoding is also realized by the most of the coordinate-wise encodings from some natural domain, as it can be concluded in the one-dimensional case.

The first step towards the results for ``typical'' encodings in the multivariate case was done in \cite{MK10}, where the authors proved that the proportion of encodings of a two-dimensional random vector that realizes a given convolution goes to one doubly exponentially. It is also explained there that the encodings behave well, not only when they are applied on i.i.d. copies of some random vector, but also when we apply them on any vector that is drawn from bi-variate (strictly stationary) ergodic source.

Our work presented in this article extends the control on the entropy-profile of transformed variables for the general multivariate case whenever the original source possesses asymptotic equipartition property (AEP). Our main results are stated in Theorems \ref{thm:encodings_convolution}, \ref{thm:encodings_convolution_sequence} and \ref{thm:encodings_convolution_erg}. In Theorem \ref{thm:encodings_convolution}, we introduce an explicit lower bound of the proportion of encodings that transform a multivariate random variable into a variable with the entropy profile close to the suitable convolution. The bound is given in terms of the entropy of the original random variable and works when the mean fluctuation of information functions is small. We use the bound to develop an asymptotic scheme in Theorem \ref{thm:encodings_convolution_sequence} that is applied to get the result on typical encodings for ergodic processes (Theorem \ref{thm:encodings_convolution_erg}). Last but not least, we control not only the entropy profile of the transformed variables, but we show that they also possess some kind of equipartition property.
To describe the equipartition property of a random variable, we introduce a new quantity that measures the non-uniformness in a way that is well preserved via transformations (encodings), conditioning, and i.i.d. expansions, namely the mean fluctuation of the information functions, see Section \ref{sec:equidistr} (and Section \ref{sec:mean fluctuation} for more details).

Let us stress out that the extraction of the critical property, namely the asymptotic equipartition property, which is sufficient assumption in Theorem \ref{thm:encodings_convolution}, allows us to extend the previous works on this topics in two significant ways; the source needs to be neither i.i.d., nor stationary. It is satisfactory if the original process is asymptotically mean stationary with ergodic mean, as defined in \cite[page 16]{gray2011entropy} (the mentioned result can be found therein as Theorem 4.1 and Section 4.5), e.g., finite-state Markov chains of any order, its functions, block codings of stationary processes, etc. For the same reason, our results can be extended to the situation when one considers a family of ergodic random fields (corresponding with an action of an amenable group, see \cite{Kie75}) instead of a family of random processes. In Section \ref{sec:ergodic}, we introduce an example of a non-stationary Markov chain and its encodings in a binary alphabet that shows the generality of our results.

We generalize the known results also in another important direction. Our results also cover the situation when only some coordinates of the multivariate source are encoded. In particular, the theorem can be used to describe the common entropy profile of the family that consists of the original variables as well as the encoded ones.

\section{Equipartition property and mean fluctuation of the information function}
\label{sec:equidistr}

Let us recall some basic notions from Information Theory. Let $\PP$ be a probability on a finite set $\xc$ (not necessarily a subset of real or complex domain). The information function $\ic_\PP:\xc\to\RR$ is given by the formula $\ic_\PP(x)=-\ln \PP(x)$. The entropy $H(\PP)$ is its expectation, i.e. $H(\PP)=\sum\PP(x)(-\ln \PP(x))$, where we sum over all $x\in \xc$ of positive probability. The set of all $x\in \xc$ of positive probability is the support of $\PP$, denoted by $s(\PP)$. A discrete finite-valued random variable $\RV$, e.g. a measurable map from a probability space $(\Omega,\PP)$ into a finite set $\xc$, induces in a natural way a discrete probability measure $\PP_X$ on $\xc$, so we can extend immediately the previous notions, the information function, the entropy and the support, as follows:
$$s(X):=s(\PP_X), \ic_X:=\ic_{\PP_X}, H(X):=H(\PP_X).$$
We will often use in the text this small abuse of notation when the random variable is written down instead of the induced probability.  

Let $\RP=(\RV(n))_{n\in\NN}$ be a random process with values in a finite alphabet $\ac$. For given $n$, $\RV\un=(\RV(i))^n_{i=1}$ is understood as a random variable with values in $\ac^n$.  The entropy rate of the process $\RP$ is defined by the formula
$$h(\RP)=\lim_{n\to\infty}\frac1n H(\RV\un).$$

The asymptotic equipartition property (AEP) claims that the entropy rate is well defined, the limit above exists, and it is equal to the limit of $\frac1n \ic_{\RV\un}$ with probability one (it is well known that i.i.d. processes, ergodic processes posses AEP, see \cite{cover2012elements}). In particular, the AEP claims that $\frac1n \ic_{\RV\un}$ is very "flat''. In order to describe such behavior in an efficient manner, we introduce the following quantities for a discrete probability $\PP$ and a real value $a\in\RR$:
\begin{align*}
  \Mean(\PP,a)&=\EE_{\PP}\left|\ic_\PP-a \right| &, \Mean(\PP)&=\Mean(\PP,H(\PP)).\\
\end{align*}
We call $\Mean(\PP,a)$ the mean fluctuation of the information function from $a$ and $\Mean(\PP)$ simply the mean fluctuation. We again extend these definitions for random variables with finite values in the natural way,  $M(X,a):=M(\PP_X,a)$ and $M(X):=M(\PP_X)$.

\section{Multivariate Random Variables}
\label{sec:multivariate}

In order to study the multivariate random variables, we admit that a random variable $X$ has its own structure, namely, that $X=(X_i)^k_{i=1}$, where $X_i$ is a random variable with values in a finite alphabet $\ac_i$, $i\leq k$. First, we introduce necessary notations. Let us denote the set $\{1,2,\ldots,k\}$ by $\hat k$ and the set of all subsets of $\hat k$ by $\tilde k$ (the same convention is used for other natural numbers). The entropy profile of $X$ is the point $\vec{H}(X)\in\RR^{\tilde k}$ given by the formula, $\vec{H}(X)=(H(X_I))_{I\in{\tilde k}}$, where $X_I$ denotes the sub-vector $(X_i)_{i\in I}$ ($H(X_\emptyset)$ is defined to be zero). In a consistent manner, we put $M(X_\emptyset)=0$ and define the maximal mean fluctuation
$$M'(X)=\max_{I\in\tilde k}M((X_i)_{i\in I}).$$

We will consider coordinate-wise encodings that encode all or only some of the coordinates. For this generality, $\ell\leq k$ is specified, as well as the family $\bc=(\bc_i)^{\ell}_{i=1}$ of finite output alphabets. 

A mapping $f$ from $\prod^k_{i=1}\ac_i$ to $\prod^\ell_{i=1}\bc_i\times \prod^k_{k=\ell+1}\ac_i$ is a coordinate-wise encodings of the first $\ell$ coordinates if it satisfies the formula 

$$f(x_1,\ldots,x_k)=(f_1(x_1), f_2(x_2),\ldots, f_\ell(x_\ell), x_{\ell+1}, x_{\ell+2},\dots, x_{k})$$
for some family of mappings $(f_i)_{i\leq \ell}$, where $f_i:\ac_i\to\bc_i$, $i\leq\ell$.

Since $f$ is determined by $(f_i)_{i\leq \ell}$, we identify the mapping with the family, i.e. we write $f=(f_i)_{i\leq\ell}$ . The set of all these coordinate-wise encodings of the first $\ell$ coordinates is denoted by $\ec_\ell$.

We define also the entropy profile $\vec{H}(\bc)\in\RR^{\tilde\ell}$ of the output alphabets as follows
$$(\vec{H}(\bc))_I=\ln |\prod_{i\in I}\bc_i|,\qquad I\in\tilde\ell.$$

We are interested in the question how the encodings change the entropy profile, i.e. what we can say about $\vec{H}(f(X))=H(f_I(X_I))_{I\in\tilde k}$. It is quite straightforward to show that the profile is coordinate-wise bounded by the convolution $\vec{H}(X)*\vec{H}(\bc)$. In general, convolution $w=u*v$ of two points $u\in\RR^{\tilde k}$ and $v\in\RR^{\tilde \ell}$, $\ell\leq k$, is the point from $\RR^{\tilde k}$ defined by the formula
$$(u*v)_I=\min_{J\subset I\cap\hat\ell}(u_{I\setminus J}+v_J),\qquad I\in\tilde k.$$

\begin{pro}\label{pro:convolution_as_upper-bound}
  Let $f$ be an encoding from $\ec_\ell$. Then 
  $$H(f(X))_I\leq(\vec{H}(X)*\vec{H}(\bc))_I,\qquad I\in\tilde k.$$
  \end{pro}
  \begin{proof}
    By the definition of the convolution, it is enough to prove that
    $$H(f_I(X_I))\leq H(X_{I\setminus J})+(\vec{H}(\bc))_J,$$
    for all $I\subset\tilde k$ and $J\subset I\cap\tilde\ell$. But $H(f_I(X_I))$ is bounded from above by the sum of $H(f_{I\setminus J}(X_{I\setminus J}))$ and $H(f_J(X_J))$, where the former entropy is surely bounded by the entropy of the source $H(X_{I\setminus J})$ and the latter by the logarithm of the cardinality of the output set $\prod_{i\in J}\bc_J$.  
  \end{proof}

In the next theorem, we show much more, namely that for a large part of encodings, the entropy profile $\vec{H}(f(X))$ is not just bounded by the convolution, but it is close to this bound. We also show that the maximal mean fluctuation can be very small at the same moment. The proof of Theorem \ref{thm:encodings_convolution} is postponed to the last section.

\begin{thm}\label{thm:encodings_convolution}
  Let $1\leq k$, $1\geq\eps>0$, $\ell\leq k$, $\delta=\left(\frac{\eps}{121}\right)^{2^{|\ell|}}$, $H>0$ and $X=(X_i)_{i\leq k}$ be a family of discrete random variables such that $X_i$ takes values in $\ac_i$, $i\leq k$, and
    \begin{align}\label{ineq:HbigI}
H>H(X_\ell),\qquad H\geq \frac{2\ln 2}{\delta}, \qquad H\geq \frac{M'(X)}{\delta}.
  \end{align}
  The proportion of those encodings $f\in\ec_\ell$ that satisfy the conditions
    \begin{align}\label{ineq:MHclose}
M'(f(X))\leq \eps H\qquad \&\qquad \left\lVert H(f(X))-\vec{H}(X)*\vec{H}(\bc)\right\rVert_{\max} \leq \eps H,
  \end{align}
is at least
$$1-|\ell|2^{k-1}\exp\left(-\frac{\ln 2}{2} e^{\delta H}+(\vec{H}(\bc))_{\hat \ell}+2H\right).$$
\end{thm}

Let us notice that for a fixed dimension $k$, the bound for the proportion of the encodings in the theorem goes to one very fast ("doubly exponentially") with respect to $H$, provided $H$ goes to infinity, and $M'/H$ goes to zero. In the next section, we apply this idea and the theorem in the situation when an ergodic source and an a.m.s. source is encoded.

Another important remark is that we focus on the encodings that realized not only the entropy close to the convolution but also the variable with very small mean fluctuation of the information function. In other words, we ask the encoding to provide the output with some kind of equipartition property. This is essential to build an inductive proof. After an application of the theorem along with the induction, the intermediate random variables we get after the encoding satisfies the assumptions of the theorem again.

\section{Asymptotic scheme}
In this section, instead of encodings of one family of random variables, we will consider a sequence of families and their encodings to different alphabets. Our aim is to construct an asymptotic scheme that is presented in Theorem \ref{thm:encodings_convolution_sequence}.

We fix $\ell\leq k$. For given $n\geq 1$, we consider a family of random variables $X\un=(X\un_i)_{i\leq k}$ defined on the same probability space, where $X\un_i$ takes values in a finite set $\ac\un_i$. Put $\ac\un=\prod_{i\leq k}\ac\un_i$. As well as in the previous section, we fix a family of finite sets $\bc\un=(\bc\un_i)_{i\leq\ell}$ and denote by $\ec\un_\ell$ the set of all mappings from $\ac\un$ to $\bc\un$ of the form  
$$f(x_1,\ldots,x_k)=(f_1(x_1), f_2(x_2),\ldots, f_\ell(x_\ell),x_{\ell+1},x_{\ell+2},\ldots,x_{k}),$$
for some $f_i:\ac\un_i\to\bc\un_i$, $i\leq \ell$. Let us recall, that we call these mappings coordinate-wise encodings of the first $\ell$ coordinates.

\begin{thm}\label{thm:encodings_convolution_sequence}
 Let $\frac{\vec{H}(X\un)}{n}$ converges to a non-zero $h\in\RR^{\tilde k}$, $\frac{M'(X\un)}{n}$ tends to zero and $\frac{\vec{H}(\bc\un)}n$ converges to $b\in\RR^{\tilde\ell}$.
 
 If $1\geq\eps>0$, $\delta<\left(\frac{\min(\eps,h_{\hat k})}{121 h_{\hat k}}\right)^{2^{|\ell|}}$, $n$ large enough, then the proportion of those encodings $f\in\ec_\ell\un$ that satisfy the conditions
    \begin{align*}
\frac{M'(f(X\un))}{n}\leq \eps\qquad \&\qquad \left\lVert \frac{\vec{H}(f(X\un))}n-h*b\right\rVert_{\max} \leq \eps,
  \end{align*}
is at least
$$1-\exp\left(- e^{\delta n}\right).$$
\end{thm}

The proof is postponed to the last section.

We developed the asymptotic scheme in the case when the limits of some numerical characteristics are assumed to exist. Nevertheless, the scheme does not require any structural relation between $X\un$ and $X^{(n+1)}$. In the next section, we will apply this scheme in the case when $X\un$ arises as the first $n$-coordinates of some process $(\RV(n))_{n\in\NN}$ and where $X^{(n+1)}$ contains $X\un$ as its beginning.

\section{Encodings of Ergodic processes and a.m.s. processes with ergodic mean}
\label{sec:ergodic}

Let $\RP=(\RV(n))_{n\in\NN}$ be a multivariate random process with values in a Cantor product of finite sets $\ac_i$, $1\leq i\leq k$. Put $\ac=\prod^k_{i=1}\ac_i$. Hence, each $\RV(n)$ is a tuple of random variables, $\RV(n)=(\RV_i(n))^k_{i=1}$. For a subset of coordinates $J\subset\hat k$, we define a sub-process $\RP_J=(\RV_J(n))_{n\in\NN}$ in the following way: $\RV_J(n)=(\RV_j(n))_{j\in J}$. As in the previous sections, $X\un$ stands for the vector $(X(1),X(2),\ldots,X(n))$, $X_J\un$ stands for $(X_J(1),X_J(2),\ldots,X_J(n))$.

We define an entropy profile of the multivariate process $\RP$ as the vector $\vec h(\RP)=(h(\RP_J))_{J\in\tilde k}$, where $h(\RP_J)$ is the entropy rate of the process $\RP_J$, i.e.
$$h(\RP_J)=\lim_{n\to\infty}\frac1n H(\RV_J(1),\RV_J(2),\ldots,\RV_J(n)),\qquad J\subset\hat k.$$

There is a quite large class of processes for which the entropy rates are well defined and $M(X\un_J)/n$ goes to zero for every $J\subset\hat k$. In order to explain this class we need to assign a process with the corresponding measure on the output-sequences. Namely, $\RP=(\RV(n))_{n\in\NN}$ with values in a finite set $A$ gives rise the measure $\RP_*\PP$ on $A^{\NN}$ that is determined by the equalities
$$\RP_*\PP([a_1\ldots a_n])=\PP(X(0)=a_0,\ldots,X(n)=a_n),\qquad n\in\NN, a_1,\ldots a_n\in A,$$
where
$$[a_1\ldots a_n]=\{(x_i)_{i\in\NN}\in A^{\NN}\mid x_i=a_i, \text{ for }i\leq n\},\qquad n\in\NN, a_1,\ldots a_n\in A.$$
The measure is defined on the $\sigma$-field $\fc$ generated by the above-mentioned sets that are usually called cylinders.

We define the shift-map $T$ on $A^\NN$ by the formula $T(x_1x_2...)=(x_2x_3...)$. A probability measure $\mu$ on $\fc$ is
\begin{itemize}
\item \emph{asymptotically mean stationary (a.m.s.)} if $\frac1n\sum^n_{i=1} \mu(T^{-i}F)$ converges, for every $F\in\fc$,
\item \emph{stationary} if $\mu(T^{-1}F)=\mu(F)$ for every $F\in\fc$.
\item \emph{ergodic} if it is stationary and $T^{-1}F=F$ and $F\in\fc$ implies $\mu(F)=\{0,1\}$.
\end{itemize}
We say that a process is a.m.s., stationary or ergodic, if the measure $\RP_*\PP$ has the corresponding property. If a process is a.m.s., then the formula
$$\RP^m_*\PP=\lim_{n\to\infty}\frac1n\sum^n_{i=1} \RP_*\PP(T^{-i}F), \qquad F\in\fc,$$  
defined a probability stationary measure on $\fc$ (the upper index "m" stands for "mean"). This measure is called the mean of the process. We will be interested in the a.m.s. processes with ergodic mean. The following theorem is a slightly weaker version of Corollary 4 in \cite{10.2307/2242939} translated into our notations and settings.

\begin{pro}[\cite{10.2307/2242939}]
  Let $\RP$ be an a.m.s. process with ergodic mean, $\mathbb{Y}$ be a stationary process with the same mean. Then the entropy rate $h(\RP)$ is well-defined and equal to $h(\mathbb Y)$. In addition, $\frac1n M(\RP^{(n)})$ goes to zero.
\end{pro}

\begin{cor}\label{cor:entropy-profile-ams}
  Let $\RP$ be an a.m.s. process with ergodic mean, $\mathbb{Y}$ be a stationary process with the same mean. Then the profile $\vec h(\RP)$ is well-defined and equal to $\vec h(\mathbb Y)$. In addition, $\frac1n M'(\RP_J^{(n)})$ goes to zero for every $J\subset\hat k$.
\end{cor}

\begin{proof}
  For $J\subset\hat k$, the natural projection from $(\prod^k_{i=1}A_i)^\NN$ onto $(\prod_{i\in J} A_i)^\NN$ intertwines with the shift-map on both spaces, so it is a factor mapping in the category of dynamical systems. In addition, the projection maps $(\RP)_*\PP$ onto $(\RP_J)_*\PP$ and the a.m.s. property is preserved via the factor mapping. So $(\RP_J)*\PP$ is a.m.s. In particular, the entropy rate $h(\RP_J)$ is well defined and $\frac1n M(\RP_J^{(n)})$ goes to zero. 
\end{proof}

The following theorem is a straightforward consequence of Corollary \ref{cor:entropy-profile-ams} and Theorem \ref{thm:encodings_convolution}.

\begin{thm}\label{thm:encodings_convolution_erg}
  Let $\RP=(\RV(n))_{n\in\NN}$ be a.m.s. with ergodic mean, $h$ be its entropy-rate profile.  If $h_{\hat k}>0$, $1\geq\eps>0$, $\delta<\left(\frac{\min(\eps,h_{\hat k})}{121 h_{\hat k}}\right)^{2^{|\ell|}}$ and $n$ large enough, then the proportion of those encodings $f\in\ec\un_\ell$ that satisfy the conditions
    \begin{align*}
\frac{M'(f(X\un))}{n}\leq \eps\qquad \&\qquad \left\lVert \frac{\vec{H}(f(X\un))}n-h*b\right\rVert_{\max} \leq \eps,
  \end{align*}
is at least
$$1-\exp\left(- e^{\delta n}\right).$$
\end{thm}

Let us point out that the class of a.m.s. processes with ergodic mean covers all ergodic processes (e.g., i.i.d. processes). It also contains all irreducible (possibly periodic) finite-states Markov chains. Other examples of a.m.s. processes can be found in \cite{10.2307/2242939}, below Corollary 4.

At the end of the section, we will exhibit the generality of the theory applying previous theorem and corollary to encodings of a non-stationary and non-independent, but Markov process.

Put $k=\ell=2$, $A_1=A_2=\ZZ_8$. The chain $\RP$ is defined as the random walk on $\ZZ_8\times\ZZ_8$ (a screwed chessboard), where we can move only one step to in the horizontal direction or a one step in the vertical direction. Since $7$ plus $1$ is $0$, we suppose that $7$ and $0$ are adjacent values. Namely, the transition probabilities are given by the following formula:
\begin{align*}
  p_{(i,j)(i',j')}&=
  \begin{cases}
    \frac14,& \text{  if } i=i'\text{ and }|j-j'| \in \{1,n-1\},\\
    \frac14,& \text{  if } j=j'\text{ and }|i-i'| \in \{1,n-1\},\\
    0, &{ otherwise.}
\end{cases}
\end{align*}
Let us fix a deterministic start at the origin, $X(0)=(0,0)$. By the standard analysis of homogeneous Markov chains with finite states we get that the chain is irreducible, periodic with period two, and non-stationary because the initial distribution is not equal to the stationary one. It is straightforward that the stationary distribution is the uniform distribution. From every states, there are four equiprobable ways out. This leads to the fact that $H(X(n+1)|X(n))$ equals to $2\ln 2$ (for any initial distribution). If we focus on the first coordinate of the process, one can see that in the next step, we can increase the value by one (modulo n) with probability $1/4$, decrease the value by one (modulo n) with probability $1/4$ or stay at the same value with probability $1/2$. In particular, $H(X_1(n+1)|X_1(n))$ is equal to $\frac{3}{2}\ln 2$. The same is true for the second coordinate. By homogeneity, the entropy rates $\frac1n H(X\un)$, $\frac1n H(X\un_1)$ and $\frac1n H(X\un_2)$ converge to the mentioned conditional entropies $2\ln 2$, $\frac32\ln 2$ and $\frac32\ln 2$, respectively. Hence, $\frac1n\vec H(X\un)$ converges to a non-zero $h\in\RR^{\tilde k}$, where
$$h:=(h_{\emptyset}, h_1, h_2, h_{1,2})=(0,\frac32\ln 2,\frac32\ln 2,2\ln 2).$$

In addition, let us assume that for encoding of the first $n$ moves in the screwed chessboard, we use $2^n$ colors for vertical position, as well as, for horizontal position. The alphabet $\bc_i\un$, $i=1,2$, can be understood as the set of all binary strings of the length $n$. In the terms of the entropy,
$$\frac1n\left(\vec{H}(\bc\un)_\emptyset,\vec{H}(\bc\un)_1,\vec{H}(\bc\un)_2,,\vec{H}(\bc\un)_{1,2} \right)=(0,\ln 2,\ln 2,2\ln 2).$$

Applying Theorem \ref{thm:encodings_convolution_erg}, we can say that a typical pair of encodings $f=(f_1,f_2)$, $f_1:(\ZZ_8)^n\to 2^n$ and $f_1:(\ZZ_8)^n\to 2^n$, yields the transformations of $X\un$ with the entropies satisfying:
\begin{align*}
  \frac1n \vec{H}(f(X\un))\sim (0,\frac32\ln 2,\frac32\ln 2,2\ln 2)*(0,\ln 2,\ln 2,2\ln 2)=(0,\ln 2,\ln 2,2\ln 2).
\end{align*}

Let us say, that for the evaluation of $\vec{H}(X_i\un)$, $i=1,2$, it was useful, that both processes, $X_i\un$ and $X_2\un$ were Markov. In the next example we relax this property.

Let us now consider a slight variation of the previous example, namely the random walk $\YY$ on the standard chess board. So the values $0$ and $7$ are not adjacent any more. We say that two elements $(i,j)$ and $(i',j')$ from $\ZZ_8\times\ZZ_8$ are adjacent if they are adjacent in one coordinate and equal in the other, i.e. if the sum of differences $|i-i'|$ and $|j-j'|$ equals one. We denote by $V_{i,j}$ the number of pairs from $\ZZ_8\times\ZZ_8$ adjacent to $(i,j)$ and define the transition probabilities as follows,   
\begin{align*}
  p_{(i,j)(i',j')}&=
  \begin{cases}
    \frac1{V_{i,j}},& \text{  if } (i,j)\text{ and } (i',j')\text{ are adjacent},\\
    0, &{ otherwise.}
\end{cases}
\end{align*}
Let us fix a deterministic start at the origin, $Y(0)=(0,0)$. Again, this Markov chain with finite states is homogeneous, irreducible, periodic with period two and non-stationary because the initial distribution is not equal to the stationary one. In order to find the entropy profile of the Markov chain, it is very handful to replace the process by its stationary version that has the same profile due to Corollary \ref{cor:entropy-profile-ams}. It helps to establish the entropy rate of the whole process. Nevertheless, the process $(Y_1(n))_{n\in\NN}$ and $(Y_2(n))_{n\in\NN}$ are not Markov, so their entropy rate can be only estimated. Using appropriate formula for the entropy rate of a Markov chain and
a lower bound for the entropy rate of a function of a Markov chain, we get that $h(\YY)$ is around $1.83\ln 2$ and $h(\YY_1)=h(\YY_2)>1.29\ln 2$. But even very simple observations lead to the facts that all the rates $h(\YY)$, $h(\YY_1)$ and $h(\YY_2)$ belong into the interval $(\ln 2,2\ln 2)$, what is important to evaluate the convolution below.  

If we use an encoding schemes into the same alphabets as in the previous example, we obtain that a typical pair of encodings $f=(f_1,f_2)$, $f_1:(\ZZ_8)^n\to 2^n$ and $f_1:(\ZZ_8)^n\to 2^n$, yields the transformations of $Y\un$ with the entropies satisfying:
\begin{align*}
  \frac1n \vec{H}(f(Y\un))\quad\sim\quad\vec h(\YY)*(0,\ln 2,\ln 2,2\ln 2)=(0,\ln 2,\ln 2,h(\YY))=(0,1,1,1.83)\ln 2.
\end{align*}

\section{Mean fluctuation of the information function}
\label{sec:mean fluctuation}

As far as we know, the mean fluctuation $M(\RP)$ of the information function $\Info_\RP$ from its mean $H(\RP)$ is not established explicitly in the literature. Nevertheless, we found it very efficient and elegant to use it as a quantity that helps to describe asymptotic equipartition property. In this section, we would like to introduce some properties of the mean fluctuation that are interested in its own.

Let us introduce some similar quantities for a probability measure on a finite set, 
\begin{align*}
  \Mean^+(\PP,a)&=\EE_{\PP}\left(\ic_\PP-a \right)^+, & \Div(\PP)&=\Mean(\PP,\ln(\#s(\PP))),\\
  \Mean^-(\PP,a)&=\EE_{\PP}\left(\ic_\PP-a \right)^-, & \Div^+(\PP)&=\Mean^+(\PP,\ln(\#s(\PP))),
\end{align*}
where the notation $a^+$ and $a^-$ stands for positive and negative parts of a number $a$, respectively. In the similar manner, we define $D^-$. We can express the entropy and the divergence from the uniform distribution on the support of the measure as follows:
$$H(\PP)=\Mean(\PP,0),\qquad \Divkl(\PP||\Unif(s(\PP)))=\Div(\PP)-2\Div^+(\PP).$$
Let us point out that the above-mentioned Kullback-Leibler divergence from the uniform distribution, as well as the mean fluctuation from $a=\ln(\#s(\PP))$, is well defined only if the support $s(\PP)$ is finite, whereas the mean fluctuation does not need the finiteness ( we assumed the finiteness of probability spaces only for simplicity).

For a discrete random variable $\RV$, we extend the previous notions in a natural way, $M(X,a):=M(\PP_X,a)$,  etc.

Let us notice that the term $\Div^+(X)$ is bounded above by $\frac{\ln e}e$ and can be often neglected as a term of small magnitude with respect to $D(X)$. Using the notation $\hat s(\RV)$ for the subset of the support of $\PP_\RV$ that contains very small atoms, i.e. the values whose probability is less than $1/(\#s(\RV))$, we get the mentioned bound:
\begin{align*}
 \Div^+(\RV)&=\sum_{x\in\hat s(\RV)}\PP(x)\ln\frac1{\#s(\RV)\PP(x)}\\
  &=\PP(\hat s(\RV))\sum_{x\in \hat s(\RV)}\frac{\PP(x)}{\PP(\hat s(\RV))}\ln\frac1{\#s(\RV)\PP(x)}\\
                                      &\leq \PP(\hat s(\RV))\ln\sum_{x\in \hat s(\RV)}\frac1{\PP(\hat s(\RV))\#s(\RV)}\\
                                      &\leq \PP(\hat s(\RV))\ln\frac1{\PP(\hat s(\RV))}\leq \frac{\ln e}{e}.
\end{align*}

For a random variable, we define the following relative versions of the mean fluctuations:
\begin{align*}
  \Mr(\RV)=\frac{\Mean(\RV)}{H(\RV)},\qquad \Dr(\RV)=\frac{\Div(\RV)}{H(\RV)}.
\end{align*}

The definition is correct as $H(\RV)>0$. Otherwise, $\Mr $ and $\Dr$ are set to be zero. We call $\Mr$ the relative mean fluctuation and $\Dr$ the relative index of uniformity. Let us recall that for a positive random variable $\eta$ the mean fluctuation is bounded as follows:
\begin{align*}
\EE\left|\EE(\eta)-\eta\right|=2\EE\left(\EE(\eta)-\eta\right)^+\leq 2\EE(\eta). 
\end{align*}
Hence, $\Mr$ is bounded by 2, whereas $\Dr$ has no reasonable bound. The following lemma shows that $\Dr$ dominates $\Mr$.
\begin{lem}
If $H(\RV)>0$, then 
$$\Mr(\RV)\leq 2\Dr(\RV).$$
\end{lem}
\begin{proof}
 Obviously,
\begin{align*}
  \Mean(X,H(X))&< \Mean(X,\ln(\#s(X)))+|H(X)-\ln(\#s(X))|\leq 2D(X).
 \end{align*}
 \end{proof}

  Since $H(\PP)$ is the expectation of $\ic_{\PP}$, we get
 \begin{align*}
   M^-(\PP)=M^+(\PP),\quad M(\PP)=2M^-(\PP)=2M^+(\PP).
 \end{align*}
 
 \begin{lem}\label{lem:error_conv_comb_I}
  Let $\PP=(1-\eps)\PP'+\eps\PP''$ for three discrete probability measures $\PP$, $\PP'$ and $\PP''$ defined on the same space and $\eps\in (0,1)$. Then
  \begin{align*}
    M(\PP)&\leq2\eps(H(\PP'')+2H(\PP'))+2M(\PP')+10\ln 2.\\
  \end{align*}
\end{lem}
\begin{proof}
    Let us recall that
  \begin{align*}
    (1-\eps)H(\PP')+\eps H(\PP'')\leq H(\PP)\leq (1-\eps)H(\PP')+\eps H(\PP'')+1.
  \end{align*}
  Let $A$ denotes the set of all $x$'s such that $(1-\eps)\PP'(x)>\eps\PP''(x)$. We conclude the proof by the following calculation:
  \begin{align*}
    \frac12 M(\PP)
    &=M^-(\PP)=\sum_x\PP(x)\left(H(\PP)+\ln \PP(x)\right)^+\\  
    &\leq\sum_x\PP(x)\left((1-\eps)H(\PP')+\eps H(\PP'')+\ln 2+\ln \PP(x)\right)^+\\  
    &\leq\eps H(\PP'')+\ln 2+\sum_{x\not\in A}2\eps\PP''(x)\left(H(\PP')+\ln 2\eps\PP''(x)\right)^+\\  
    &\quad+\sum_{x\in A}2(1-\eps)\PP'(x)\left(H(\PP')+\ln 2(1-\eps)\PP'(x)\right)^+\\  
    &\leq\eps H(\PP'')+\ln 2+2\eps(H(\PP')+\ln 2)+2\ln 2+2M^-(\PP').\\
  \end{align*}
\end{proof}

\begin{lem}\label{lem:error_conv_comb_II}
  Let $\PP=(1-\eps)\PP'+\eps\PP''$ for three discrete probability measures $\PP$, $\PP'$ and $\PP''$ defined on the same space and $\eps\in (0,1)$. Then
  \begin{align*}
    M(\PP)&\leq2\left(\eps H(\PP)+\ln2+\sum_x\PP(x)\left(H(\PP)-\Info_{\PP'}(x)\right)^+\right).\\
  \end{align*}
\end{lem}
\begin{proof}
    Let us recall that
  \begin{align*}
    (1-\eps)H(\PP')+\eps H(\PP'')\leq H(\PP)\leq (1-\eps)H(\PP')+\eps H(\PP'')+1.
  \end{align*}
  Let $A$ denotes the set of all $x$'s such that $(1-\eps)\PP'(x)>\eps\PP''(x)$. We conclude the proof by the following calculation:
  \begin{align*}
    \frac12 M(\PP)
    &=M^-(\PP)=\sum_x\PP(x)\left(H(\PP)+\ln \PP(x)\right)^+\\  
    &=\ln2+ \sum_{x\not\in A}\PP(x)\left(H(\PP)+\ln \eps\PP''(x)\right)^+\\  
    &\quad+\sum_{x\in A}\PP(x)\left(H(\PP)+\ln(1-\eps)\PP'(x)\right)^+\\  
    &\leq\ln2+\eps H(\PP)+\sum_{x\in A}\PP(x)\left(H(\PP)+\ln\PP'(x)\right)^+.\\
  \end{align*}
\end{proof}

We already mentioned that for a variable $X$ with infinite support $s(X)$, the value $\Div(X)$ and $\Dr(X)$ is not well-defined whereas the values $M(X)$ and $\Mr(X)$ can take arbitrarily small positive values. In this section, we show that the difference between these two notions remains significant even in the case of finite-valued i.i.d. process.

The definition of $M_{rel}$ is introduced in the beginning of Section \ref{sec:proofs}.
\begin{lem}\label{lem:mean-err-for-ergodic}
  Let $\RP=(\RV_i)_{i\in\NN}$ be an ergodic stationary process with strictly positive and finite entropy rate $h$. Then 
$$\lim_{n\to\infty}\Mr(\RV_1,\RV_2,\ldots,\RV_n)=0.$$
\end{lem}
  
We introduce also the conditional counterpart that covers the previous lemma by putting $Y_i$ to be a constant.
\begin{lem}
  Let $\RP=(\RV_i,Y_i)_{i\in\NN}$ be an ergodic stationary process with strictly positive (and finite) entropy rate $h(X|Y)$. Then 
$$\lim_{n\to\infty}\Mr(\RV^n_1|Y^n_1)=0.$$
\end{lem}
\begin{proof}
  We use the weaker form of conditional AEP for the ergodic processes. Namely, $\frac1n \Info_{\RV^n_1|Y^n_1}$ converges to $h$ in probability. Since $\frac1n H(\RV^n_1|Y^n_1)$ goes to $h(X|Y)$ too, the difference
  $$\xi_n=\frac1n \left(\Info_{\RV^n_1|Y^n_1}-H(\RV^n_1|Y^n_1)\right)$$
  converges to zero in probability. Since $\EE\xi_n$ is zero,
  $$\EE|\xi_n|=2\EE\xi^-_n.$$
  But $\xi^-_n$ goes to zero in $\lc_1$-norm, because it converges in probability and is bounded by $H(\RV_1|Y_1)$. It follows that $\xi_n$ goes to zero in $\lc_1$-norm, i.e. $\EE|\xi_n|$ goes to zero. Thus,
  
  \begin{align*}
    \lim_{n\to\infty}\Mr(\RV^n_1|Y^n_1)&=\lim_{n\to\infty}\frac{\EE\left|\Info_{\RV^n_1|Y^n_1}-H(\RV^n_1|Y^n_1)\right|}{H(\RV^n_1|Y^n_1)}=\lim_{n\to\infty}\frac{\EE|\xi_n|}{H(\RV^n_1|Y^n_1)/n}\\
    &=\frac0{h(X|Y)}=0.
  \end{align*}
\end{proof}

The following lemma is a direct consequence of the property of the divergence for the product measures.
\begin{lem}
If $=(\RV_i)_{i\in\NN}$ is i.i.d., $s(X_1)$ is finite and $H(\RV_1)>0$, then $\Dr(\RV^n_1)$ converges to $\frac{\ln(\#s(\RV_1))-H(\RV_1)}{H(\RV_1)}$. In particular, the sequence $\Dr(\RV^n_1)$ converge to zero if and only if every $X_i$ is uniform.  
\end{lem}

The lemmas show that the control of the uniformity of the distribution of the random variable via the relative mean fluctuation is weaker than that via the relative divergence from the uniform distribution on the set of values.

\section{Proofs of main theorems}
\label{sec:proofs}
In this section, we prove Theorems ~\ref{thm:encodings_convolution} and \ref{thm:encodings_convolution_sequence}. The section is self-contained, and the lemmas from the previous section are not involved. Proposition \ref{pro:number_of_encodings_conditional} is proved with several free parameters where the full generality is aimed to serve as a flexible reference in future research. Afterward, many of the parameters are fixed in Proposition \ref{pro:number_of_encodings_simplified} to get a more specific result, which is used in the proofs of the theorems.

First, let us introduce conditional counterparts of quantities defines so far. Given two discrete random variables $X$ and $Y$ defined on the same probability space with values in the countable sets $\xc$ and $\yc$, respectively, we define the conditional information function by the formula $\ic_{X|Y}=\ic_{X,Y}-\ic_Y$, where all the three functions are considered on the domain $\xc\times\yc$.

We can extend the definition of $\Mean$ and $\Mr$ to the conditional case as follows:
\begin{align*}
  \Mean(\RV|Y,a)&=\EE_{\PP_{X,Y}}\left|\ic_{X|Y}-a \right| & \Mean^+(\RV|Y,a)&=\EE_{\PP_{X,Y}}\left(\ic_{X|Y}-a \right)^+\\
  \Mean^-(\RV|Y,a)&=\EE_{\PP_{X,Y}}\left(\ic_{X|Y}-a \right)^- .& 
\end{align*}
Shorter notation $\Mean(\RV|Y)$, $\Mean^+(\RV|Y)$ and $\Mean^-(\RV|Y)$ is used when $a=H(X|Y)$. In addition, when $H(X|Y)>0$, then put 
\begin{align*}
  \Mr(X|Y)=\frac{\Mean(X|Y)}{H(X|Y)}. 
\end{align*}


Since the information function satisfies the following chain rule,
$$\Info_{X|Y}+\Info_Y=\Info_{X,Y},$$
we get
\begin{align}\label{ineq:mean_error}
  M(X|Y)&\leq\EE|\Info_{Y}-H(Y)|+\EE|\Info_{X,Y}-H(X,Y)|\leq M(Y)+M((X,Y)),\\
  M(X,Y)& \leq M(Y)+M((X|Y)).
  \end{align}

  Since $H(X|Y)$ is the expectation of $\ic_{X|Y}$, we get
 \begin{align}
   M^-(X|Y)=M^+(X|Y),\quad M(X|Y)=2M^-(X|Y)=2M^+(X|Y).
 \end{align}
  
The following lemma is a simplified version of Lemma 6 in \cite{Ma07a}. It provides the crucial bound on the probability of the colored atoms that is applied in the next proposition.

\begin{lem}[Matúš, \cite{Ma07a}]\label{lem:small_atoms}
  Let $\PP$ be a sub-probability measure on a finite set $\xc$. For $k\geq 1$, $\eps>0$, the proportion of those maps (encodings) $f$ from $\xc$ into $\hat k$ that satisfy
  $$\PP(f^{-1}(j))\leq \frac{1+\eps}{k}\qquad j\in\hat k,$$
  is at least
  $$1-ke^{-\frac{\eps}{2kq}\ln(1+\eps)},$$
  where $q=\max_{x\in\xc}\PP(x)$.
\end{lem}

\begin{pro}\label{pro:number_of_encodings_conditional}
  Let $\RV,Y$ be random variables with values in finite sets $\ac_\RV$ and $\ac_Y$, respectively. Let $\bc$ be a finite set, $t_1,t_2,\delta\in\RR^+$, $r,s\in (0,1)$. Let $\alpha=\frac1{t_1} M(X|Y)$, $\beta=\frac1{t_2} M(Y)$, $R=\min(H(X|Y),\ln |\bc|)$ and $\gamma=\alpha^{1-r}+\beta^{1-s}$. Then the proportion of those maps (encodings) $f$ from $\ac_X$ into $\bc$ that satisfy the conditions
$$M(f(X)|Y)\leq 2\gamma R+2\delta+4\ln 2,$$
$$\left|H(f(X)|Y)-R\right| \leq \gamma R+\delta+2\ln 2,$$
is at least
$$1-\exp\left(-\frac{\ln 2}{2}e^{\delta+H(X|Y)-R-\alpha^r t_1}+\ln |\bc|+H(Y)+\beta^s t_2\right).$$
\end{pro}
\begin{proof}
  Denote $R=\min(H(X|Y),\ln |\bc|)$. Surely, $H(f(X)|Y)\leq R$.

  Fix $r,s>0$ and put 
  \begin{align*}
    B=&\{y\mid
        \left|\Info_{Y}-H(Y)\right|<\beta^s t_2
        \},\\
    A=&\{(x,y)| \left|\Info_{X|Y}-H(X|Y)\right|<\alpha^rt_1\}.    
  \end{align*}
  By Markov inequality, we get $\PP_{X,Y}(A)>1-\alpha^{1-r}$ and $\PP_Y(B)>1-\beta^{1-s}$. Let $\PP'$ be the restriction of $\PP_{X,Y}$ on $A$, i.e. $\PP'$ is sub-probability measure defined as follows:
  $$\PP'(x,y)=
  \begin{cases}
  \PP_{X,Y}(x,y),\quad \text{ if }\qquad (x,y)\in A,\quad y\in B,\\
  0,\quad \text{ otherwise.}
\end{cases}
$$

Given $y\in B$, measure $Q_y(x)=\PP'(x,y)/\PP_Y(y)$ is a sub-probability measure that is bounded by $e^{-H(X|Y)+\alpha^r t_1}$. We apply Lemma \ref{lem:small_atoms}. Let $c_y$ be the proportion of encodings that satisfy the condition:
\begin{align}\label{ineq:proofII}
Q_y(f^{-1}(x'))\leq \frac1{|\bc|}+e^{\delta-R},\qquad x'\in\bc.
\end{align}
  By Lemma \ref{lem:small_atoms},
  \begin{align*}
    1-c_y
    &\leq\exp\left(-\frac{e^{\delta-R}|\bc|}{2|\bc|e^{-H(X|Y)+\alpha^r t_1}}\ln(1+e^{\delta-R}|\bc|)+\ln |\bc|\right)\\
    &\leq\exp\left(-\frac{\ln 2}{2}e^{\delta+H(X|Y)-R-\alpha^r t_1}+\ln |\bc|\right).\\
  \end{align*}
  Let $c$ be the proportion of the encodings that satisfy the above-mentioned conditions for all $y\in B$ simultaneously. Then
 \begin{align*}
   1-c
   &\leq (\# B)  \exp\left(-\frac{\ln 2}{2}e^{\delta+H(X|Y)-R-\alpha^r t_1}+\ln {|\bc|}\right)\\
   &\leq \exp\left(-\frac{\ln 2}{2}e^{\delta+H(X|Y)-R-\alpha^r t_1}+\ln {|\bc|}+H(Y)+\beta^st_2\right).\\
  \end{align*}

  For the rest of the proof, let $f:\ac_X\to\bc$ be an encodings that satisfies condition (\ref{ineq:proofII}). We denote by $\PP_{f(X),Y}$ and $\PP'_{f,Id}$ the probability measures on $\bc \times\ac_Y$ that are the images of the probabilities $\PP_{X,Y}$ and $\PP'$ under the mapping $(f,Id):\ac_X\times\ac_Y\to \bc \times\ac_Y$, i.e.
$$\PP_{f(X),Y}(x',y)=\PP_{X,Y}\left(f^{-1}(x')\times\{y\}\right),\qquad  \PP'_{f,Id}(x',y)=\PP'\left(f^{-1}(x')\times\{y\}\right).$$

Put $\PP''=\PP_{f(X),Y}-\PP'_{f,Id}$, $A'=\{(x',y)\mid \PP'_{f,Id}(x',y)> \PP''(x',y)\}$. Let us notice, that $(x',y)\in A'$ implies $\PP'_{f,Id}(x',y)$is strictly positive, so  $y\in B$.

In the following calculation, we use the fact that $\PP_{f(X),Y}\leq 2\max(\PP'',\PP'_{f,Id})$:  
\begin{align*}
  M^-(f(X)|Y,R)&\leq\sum_{(x',y)}\PP_{f(X),Y}(x',y)\left(R+\ln\frac{2\max\left(\PP'_{f,Id}(x',y),\PP''(x',y)\right)}{\PP_Y(y)}\right)^+\\
  &\leq\ln 2+\PP''(\bc\times\ac_Y) R\\
            &\quad
    +
               \sum_{(x',y)\in A',y\in B}\PP'_{f,Id}(x',y)\left(R+\ln\frac{\PP'_{f,Id}(x',y)}{\PP_Y(y)}\right)^+\\
             &\leq\ln 2+(\alpha^{1-r}+\beta^{1-s})R\\
             &\quad +\left(R+\ln\left(\frac1{|\bc|}+e^{\delta-R}\right)\right)\\
             &\leq\ln 2+(\alpha^{1-r}+\beta^{1-s})R+(\delta+\ln 2)\leq \gamma R+\delta+2\ln 2.
\end{align*}
  In addition,
  \begin{align*}
    R-H(f(X)|Y)&=\EE_{\PP_{f(X),Y}}\left(R-\ic_{f(X)|Y}\right)\leq M^-(f(X)|Y,R).
  \end{align*}

  Since $H(f(X)|Y)\leq R$, $|R-H(f(X)|Y)|$ is bounded by $\gamma R+\delta+2\ln 2$. Moreover,
  \begin{align*}
    M(f(X)|Y)&=2M^-(f(X)|Y)=2\EE_{\PP_{f(X),Y}}\left(H(f(X)|Y)-\ic_{f(X)|Y}\right)^+\\
    &\leq 2\EE_{\PP_{f(X),Y}}\left(R-\ic_{f(X)|Y}\right)^+\leq 2M^-(f(X)|Y,R).
  \end{align*}
\end{proof}

If we choose $Y$ to be trivial random variable (deterministic one), then we get the following corollary for one random variable.

\begin{cor}\label{cor:number_of_encodings_single_var}
  Let $\RV$ be random variables with values in finite set $\ac$, $\bc$ be a finite set, $t\in\RR^+$, $r\in\RR$. Let $\alpha=\frac1{t} M(X)$, $R=\min(H(X),\ln |\bc|)$. Then the proportion of those maps (encodings) $f$ from $\ac$ into $\bc$ that satisfy the conditions
$$M(f(X))\leq 2\alpha^{1-r} R+2\delta+4\ln 2,$$
$$\left|H(f(X))-R\right| \leq \alpha^{1-r} R+\delta+2\ln 2,$$
is at least
$$1-\exp\left(-\frac{\ln 2}{2}e^{\delta+H(X)-R-\alpha^r t}+\ln |\bc|\right).$$
\end{cor}

The following corollary introduces the explicit bounds for the change of the joint entropy and the mean error of the joint information function when one variable is encoded into a prescribed alphabet.

\begin{cor}\label{pro:number_of_encodings_simplified}
  Let $\RV,Y$ be random variables with values in finite sets $\ac_\RV$ and $\ac_Y$, respectively. Let $\bc$ be a finite set, $1\leq\eps>0$ and $H>0$ such that
  \begin{align}\label{ineq:HbigIII}
  H\geq H(X,Y),\qquad H\geq \frac{M(X,Y)}{\eps}, \qquad H\geq \frac{M(Y)}{\eps}, \qquad H\geq \frac{4\ln 2}{\eps}.
  \end{align}
  The proportion of those maps (encodings) $f$ from $\ac_X$ into $\bc$ that satisfy the conditions
$$M(f(X)|Y)\leq 10\sqrt{\eps}H,$$
$$\left|H(f(X)|Y)-R\right| \leq 5\sqrt{\eps} H,$$
is at least
$$1-\exp\left(-\frac{\ln 2}{2} e^{\eps H}+\ln |\bc|+2H\right),$$
where $R=\min(H(X|Y),\ln |\bc|)$.
\end{cor}

\begin{proof}
Put
$$t_1=t_2=H,\qquad r=s=\frac12,\qquad \delta=\left(\eps+\sqrt{\eps}\right)H.$$

If we define $\alpha, \beta$ and $\gamma$ as in Proposition  \ref{pro:number_of_encodings_conditional}, then the inequalities (\ref{ineq:HbigIII}) and subadditivity for $M$ (see (\ref{ineq:mean_error})) ensures that $\alpha\leq 2\eps$, $\beta\leq\eps$ and $\gamma\leq(\sqrt{2}+1)\sqrt\eps$.

  By Proposition \ref{pro:number_of_encodings_conditional}
  , the proportion of those maps (encodings) $f$ from $\ac_X$ into $\bc$ that satisfy the conditions
  \begin{align*}
M(f(X)|Y)\leq 2(\sqrt{2}+1)\sqrt\eps H+2(\eps+\sqrt\eps) H+\eps H,\\
\left|H(f(X)|Y)-\min(H(X|Y),\ln |\bc|)\right| \leq (\sqrt{2}+1)\sqrt\eps H+(\eps+\sqrt\eps) H+\eps H/2,
  \end{align*}
is at least
$$1-\exp\left(-\frac{\ln 2}{2} e^{\eps H}+\ln |\bc|+H(Y)+\sqrt{\eps}H\right).$$

Since $M(Y)\leq \eps H$ and $\eps\leq\sqrt\eps$, we get 

  \begin{align*}
2(\sqrt{2}+1)\sqrt\eps H+2(\eps+\sqrt\eps) H+\eps H\leq 10\sqrt\eps H\\
(\sqrt{2}+1)\sqrt\eps H+(\eps+\sqrt\eps) H+\eps H/2\leq 5\sqrt\eps H.
  \end{align*}

In the exponent for the bound of the proportion of the encodings,  $H(Y)+\sqrt{\eps}H$ is bounded by $2H$, since $\eps\leq 1$.
\end{proof}

\begin{proof}[Proof of Theorem ~\ref{thm:encodings_convolution}]
  We will prove the theorem by the induction over $\ell$. 
  For all the proof, fix $k,\ell,\eps,\delta,X=(X_i)_{i\leq k},H$ which satisfy the assumptions of the proposition.
  
  Assume that $\ell=0$. Then the only element $f$ from $\ec_\ell$ is identity, $f(X)=X$, $w=H(X)$. The theorem follows immediately. 

  Let $\ell\geq 1$, $\eps'=\left(\frac{\eps}{121}\right)^2$ and
$$w=\vec{H}(X)*\vec{H}((\bc_i)_{i\leq \ell}),\qquad w'=\vec{H}(X)*\vec{H}((\bc_i)_{i\leq \ell-1}).$$

The set of all $f'\in\ec_{\ell-1}$ that satisfy the conditions
    \begin{align}\label{ineq:MHclose_induction}
M'(f'(X))\leq \eps' H,\qquad \left\lVert \vec{H}(f'(X))-\vec{H}(X)*\vec{H}((\bc_i)_{i\leq \ell-1})\right\rVert_{\max} \leq \eps' H,
\end{align}
is denoted by $\gc'$. Let $\pi:\ec_\ell\to \ec_{\ell-1}$ is the projection defined by the formula $\pi(f)=(f_i)_{i\leq \ell-1}\in \ec_\ell$.

Given $f\in\ec_{\ell-1}$, $J\subset\hat k\setminus\{\ell\}$, $\gc_{f,J}$ is the set of all mappings that satisfy the conditions:
    \begin{align}\label{ineq:MHclose_induction2}
M(g(X_\ell)|f_J(X_J))&\leq 10\sqrt{\eps'} H\\
\left|H(g(X_\ell)|f_J(X_J))-\min(H(X_\ell|f_J(X_J)),\ln|\bc_\ell|)\right|&\leq 5\sqrt{\eps'} H.
\end{align}
The special form of the conditions suits the later application of Corollary \ref{pro:number_of_encodings_simplified}. 
Put $\gc_f=\bigcap_J\gc_{f,J}$, where the intersection goes over all $J\subset\hat k\setminus\{\ell\}$,
$$\gc=\{f\in\ec_\ell \mid \pi(f)\in\gc', f_\ell\in \gc_{\pi(f)}\}.$$

Let $f\in\gc$, $I\subset\hat k$. If $\ell\not\in I$, then $w_I=w'_I$, $f_I$ equals $f'_I$ and (\ref{ineq:MHclose_induction}) ensures that $M(f_I(X_I))$ and the difference $|H(f_I(X_I))-w_I|$ are both bounded by $\eps H$.
If $\ell\in I$, then
$$M(f_I(X_I))\leq M(f_\ell(X_\ell))|f'_J(X_J))+M(f'_J(X_J))\leq 10\sqrt{\eps'}H+\eps' H\leq\eps H,$$
where $J=I\setminus\{\ell\}$. By Proposition \label{pro:convolution_as_upper-bound}, $w_I\geq H(f_I(X_I))$ and 
\begin{align*}
  w_I-H(f_I(X_I))
  &=w_I-H(f_J(X_J))-H(f_\ell(X_\ell)|f_J(X_J))\\
  &=w_I-H(f_J(X_J))-\min(H(X_\ell|f_J(X_J)),\ln|\bc_\ell|)+5\sqrt{\eps'} H\\
  &=w_I-\min(H(f'_I(X_I)),H(f_J(X_J))+\ln|\bc_\ell|)+5\sqrt{\eps'} H\\
  &=w_I-\min(w'_I,w'_J+\ln|\bc_\ell|)+\eps' H+5\sqrt{\eps'} H\\
  &\leq w_I-w_I+\eps' H+5\sqrt{\eps'} H\leq\eps H.
\end{align*}

Hence, condition (\ref{ineq:MHclose}) holds for all $f\in\gc$.

In order to make statements shorter and notations more readable, we put 
$$\ell'=\ell-1,\qquad D=\exp\left(-\frac{\ln 2}{2} e^{\delta H}+(\vec{H}(\bc))_{\hat{\ell}}+2H\right).$$

Since $\delta=\left(\frac{\eps'}{121}\right)^{2^{\ell}}$ and $H(X_{\hat{\ell'}})\leq H(X_{\hat{\ell}})$, the assumption of the theorem remains true when replacing $\ell$ by $\ell'$ and $\eps$ by $\eps'$. By the inductive assumption and the fact that $\pi$ is a mapping $|\ac_\ell|^{|\bc_\ell|}$ to 1, we get
\begin{align*}
c_1:=1-\frac{|\pi^{-1}(\gc')|}{|\ec_{\ell}|}&=1-\frac{|\gc'|}{|\ec_{\ell-1}|}\leq (\ell-1)\, 2^{k-1} D.  
\end{align*}

Let $f'\in\gc'$, $J\subset\hat k\setminus\{\ell\}$. We will apply Corollary \ref{pro:number_of_encodings_simplified}, where the variables $X$ and $Y$ are understood as $X_\ell$ and $f_J(X_J)$, respectively, and $\eps$ is replaced by $\eps'$.

First we have to verify the assumptions of the proposition, namely
  \begin{align*}
  H\geq H(X_\ell,f_J(X_J)),\qquad H\geq \frac{M(X_\ell,f_J(X_J))}{\eps'}, \qquad H\geq \frac{M(f_J(X_J))}{\eps'}
  \end{align*}
and $H\geq \frac{4\ln 2}{\eps'}$. The first inequality follows from the fact that $H(f_J(X_J))\leq H(X_J)$ for every $J\in\tilde k$. The second and the third one are the immediate consequence of the fact that the numerators are bounded $M'(f(X))$. The last one follows from the inequality $\eps'>\eps$.

Applying Corollary \ref{pro:number_of_encodings_simplified},
$$c_{f',J}:=1-\frac{|\gc_{f',J}|}{|\bc_\ell|^{|\ac_\ell|}}\leq \exp\left(-\frac{\ln 2}{2} e^{\delta H}+(\vec{H}(\bc))_{\hat{\ell'}}+2H\right)\leq D.$$
Hence,
\begin{align*}
c_{f'}&:=1-\frac{|\gc_{f'}|}{|\bc_\ell|^{|\ac_\ell|}}
\leq\sum_{J\subset\hat k\setminus\{\ell\}}1-\frac{|\gc_{f',J}|}{|\bc_\ell|^{|\ac_\ell|}}\leq 2^{k-1}D.
\end{align*}

By the definition of $\gc$, there is one-to-one correspondence between its elements and pairs $(f',g)$ where $f'\in\gc$ and $g\in\gc_{f'}$, given by the equality $f_i=f'_i$, $i\leq \ell-1$, $f_\ell=g$.

Hence
\begin{align*}
1-\frac{|\gc|}{|\ec_\ell|}
&=1-\frac{\sum_{f'\in\gc'}|\gc_{f'}|}{|\ec_\ell|}
\leq\frac{|\gc'|\cdot|\bc_\ell|^{|\ac_\ell|}-\sum_{f'\in\gc'}|\gc_{f'}|}{|\ec_\ell|}\\
 &\leq\frac{|\ec_\ell|-|\gc'|\cdot|\bc_\ell|^{|\ac_\ell|}+\sum_{f'\in\gc'}|\bc_\ell|^{|\ac_\ell|}-|\gc_{f'}|}{|\ec_\ell|}\\
 &\leq 1-\frac{|\gc'|}{|\ec_{\ell-1}|}+\frac1{|\ec_{\ell-1}|}
 \sum_{f'\in\gc'}\frac{|\bc_\ell|^{|\ac_\ell|}-|\gc_{f'}|}{|\bc_\ell|^{|\ac_\ell|}}\\
 &\leq c_1+\frac1{|\ec_{\ell-1}|}\sum_{f'\in\gc'}c_{f'}\leq (\ell-1)\,2^{k-1}D+\frac{|\gc'|}{|\ec_{\ell-1}|}2^{k-1}D\leq \ell\, 2^{k-1}D.
\end{align*}

\end{proof}

Before we start to prove Theorem ~\ref{thm:encodings_convolution_sequence}, let us recall its property used in the proof.
For $u,u'\in\RR^{\tilde k}$ and $v,v'\in\RR^{\tilde \ell}$,
$$||u*v-u'*v||_{max}\leq ||u-u'||_{max},\quad ||u*v-u*v'||_{max}\leq ||v-v'||_{max}.$$

It follows from the very general fact; given two vectors of real numbers $(a_i)_{i\leq m}$ and $(b_i)_{i\leq m}$,
$$|\min_{i\leq m} a_i-\min_{i\leq m} b_i|\leq\max_{i\leq m}|a_i-b_i|.$$

\begin{proof}[Proof of Theorem \ref{thm:encodings_convolution_sequence}]

  Use the shorter notation:
  $$h\un=\frac{\vec{H}(X\un)}n,\qquad g\un=\frac{\vec{H}(f(X\un))}n,\qquad b\un=\frac{\vec{H}(\bc\un)}n.$$
  Let $\delta<\left(\frac{\min(\eps,h_{\hat k})}{121 h_{\hat k}}\right)^{2^{|\ell|}}$. There exists $\eps''<\eps'<\min(\eps,h_{\hat k})$ and $\delta'>\delta$ such that $\delta'=\left(\frac{\eps''}{121 h_{\hat k}}\right)^{2^{|\ell|}}$. By the assumptions of the theorem, for $n$ big enough, $H(X\un)$ exceeds both , $(\delta')^{-1}M'(X\un)$ and $(\delta')^{-1}4\ln 2$. By Theorem \ref{thm:encodings_convolution}, the proportion of the encodings from $\ec\un_\ell$ that satisfy
    \begin{align}\label{ineq:MHcloseII}
M'(f(X\un))\leq \frac{\eps''}{h_{\hat k}}\, nh\un_{\hat k}\qquad \&\qquad \left\lVert  ng\un-nh\un *nb\un\right\rVert_{\max} \leq \frac{\eps''}{h_{\hat k}}\, nh\un_{\hat k},
  \end{align}
is at least
$$1-\exp\left(-\frac{\ln 2}2e^{\delta' n}+nb\un_{\hat \ell}+2nh\un_{\hat k} +(k-1)\ln 2+\ln\ell\right).$$

Since $h\un$ goes to $h$ and $h_{\hat k}>0$, $\frac{h\un_{\hat k}}{h_{\hat k}}$ is smaller than $\frac{\eps'}{\eps''}$ for $n$ large enough. In such a case, condition (\ref{ineq:MHcloseII}) implies
\begin{align}\label{ineq:MHcloseIII}
\frac{M'(f(X\un))}{n}<\eps'\qquad \&\qquad \left\lVert  g\un-h\un *b\un\right\rVert_{\max} <\eps',
\end{align}

In addition,
\begin{align*}
  \left\lVert  h*b-h\un *b\un\right\rVert_{\max} &\leq \left\lVert  h*b-h\un *b\right\rVert_{\max} +\left\lVert  h\un*b-h\un *b\un\right\rVert_{\max}\\
&\leq \left\lVert  h-h\un\right\rVert_{\max} +\left\lVert  b-b\un\right\rVert_{\max}  
\end{align*}
The last two terms goes to zero. Thus, for $n$ large enough, their sum is bounded by $\eps-\eps'$ and condition (\ref{ineq:MHcloseII}) implies
$$ \left\lVert  g\un-h*b\right\rVert_{\max} <\eps.$$

It remains to prove, that the lower bound for the proportion of the encodings satisfying (\ref{ineq:MHcloseII}) is larger than $1-\exp\left(- e^{\delta n}\right)$. But in the exponent of the exponential term in the bound, there is only one exponential term $-\frac{\ln 2}{2}e^{\delta' n}$, the others are at most linear. Since $\delta<\delta'$, the term $-e^{\delta n}$ is eventually bigger than all the exponent term from the bound and
$$1-\exp\left(-\frac{\ln 2}2e^{\delta' n}+nb\un_{\hat \ell}+2nh\un_{\hat k} +(k-1)\ln 2+\ln\ell\right)>1-\exp\left(- e^{\delta n}\right).$$
\end{proof}

\section*{Acknowledgments}
We are greatly indebted to Prof. Laszlo Csirmaz for useful discussions and comments that helped to find the final shape of the presented results. We would also like to thank anonymous reviewers for their comments that greatly improved the quality and exposition of this paper.
\bibliographystyle{alpha}
\bibliography{biblio} 

\newcommand{\etalchar}[1]{$^{#1}$}
\begin{thebibliography}{BMR{\etalchar{+}}13}

\bibitem[BMR{\etalchar{+}}13]{bassoli2013network}
Riccardo Bassoli, Hugo Marques, Jose Rodriguez, Kenneth~W Shum, and Rahim
  Tafazolli.
\newblock Network coding theory: A survey.
\newblock {\em Communications Surveys \& Tutorials, IEEE}, 15(4):1950--1978,
  2013.

\bibitem[CT12]{cover2012elements}
Thomas~M Cover and Joy~A Thomas.
\newblock {\em Elements of information theory}.
\newblock John Wiley \& Sons, 2012.

\bibitem[GK80]{10.2307/2242939}
Robert~M. Gray and J.~C. Kieffer.
\newblock Asymptotically mean stationary measures.
\newblock {\em The Annals of Probability}, 8(5):962--973, 1980.

\bibitem[Gra11]{gray2011entropy}
Robert~M Gray.
\newblock {\em Entropy and information theory}.
\newblock Springer Science \& Business Media, 2011.

\bibitem[Kac12]{Ka12thesis}
Tarik Kaced.
\newblock {\em Partage de secret et th{\'e}orie algorithmique de
  l’information}.
\newblock PhD thesis, Universit{\'e} Montpellier 2, 2012.

\bibitem[Kie75]{Kie75}
John~C. Kieffer.
\newblock A generalized {S}hannon-{M}c{M}illan theorem for the action of an
  amenable group on a probability space.
\newblock {\em Ann. Probability}, 3(6):1031--1037, 1975.

\bibitem[Mat07]{Ma07a}
František Matúš.
\newblock Two constructions on limits of entropy functions.
\newblock {\em Information Theory, IEEE Transactions on}, 53(1):320--330, 2007.

\bibitem[MC16]{matuvs2016entropy}
Franti{\v{s}}ek Mat{\'u}{\v{s}} and L{\'a}szlo Csirmaz.
\newblock Entropy region and convolution.
\newblock {\em IEEE Transactions on Information Theory}, 62(11):6007--6018,
  2016.

\bibitem[MK10]{MK10}
František Matúš and Michal Kupsa.
\newblock On colorings of bivariate random sequences.
\newblock In {\em IEEE International Symposium on Information Theory -
  Proceedings}, pages 1272--1275, 2010.

\bibitem[Yeu08]{yeung2008information}
Raymond~W Yeung.
\newblock {\em Information theory and network coding}.
\newblock Springer Science \& Business Media, 2008.

\end{thebibliography}

 \end{document}